\documentclass[pdflatex, sn-mathphys]{sn-jnl}% Math and Physical Sciences Reference Style
\usepackage{comment}
\usepackage{multirow}
\usepackage{tabularx}
\usepackage{appendix}
\usepackage{algorithm}
\usepackage{algpseudocode}
\usepackage{graphicx}
\usepackage{thmtools}
\usepackage{thm-restate}
\usepackage{physics}

\jyear{2022}%

\begin{document}

\title{Policy Gradients using Variational Quantum Circuits}


\author*[1,2,3]{\fnm{André} \sur{Sequeira}}\email{andre.sequeira@inl.int}

\author[1,2,3]{\fnm{Luis} \sur{Paulo Santos}}\email{psantos@di.uminho.pt}
\equalcont{These authors contributed equally to this work.}

\author[1,2,3]{\fnm{Luis} \sur{Soares Barbosa}}\email{lsb@di.uminho.pt}
\equalcont{These authors contributed equally to this work.}

\affil*[1]{\orgdiv{Department of Informatics}, \orgname{University of Minho}, \state{Braga}, \country{Portugal}}

\affil[2]{\orgdiv{HASLab}, \orgname{INESC TEC}, \state{Braga}, \country{Portugal}}
\affil[3]{\orgname{International Nanotechnology Laboratory (INL)}, \state{Braga}, \country{Portugal}}


\abstract{Variational Quantum Circuits are being used as versatile Quantum Machine Learning models. Some empirical results exhibit an advantage in supervised and generative learning tasks. However, when applied to Reinforcement Learning, less is known. In this work, we considered a Variational Quantum Circuit composed of a low-depth hardware-efficient ansatz as the parameterized policy of a Reinforcement Learning agent. We show that an $\epsilon$-approximation of the policy gradient can be obtained using a logarithmic number of samples concerning the total number of parameters. We empirically verify that such quantum models behave similarly to typical classical neural networks used in standard benchmarking environments and quantum control, using only a fraction of the parameters. Moreover, we study the Barren Plateau phenomenon in quantum policy gradients using the Fisher information matrix spectrum.}




\maketitle

\section{Introduction}\label{sec: Introduction}
%RL is responsible for the most important steps in Artificial Intelligence nowadays. 
Reinforcement Learning (RL) is responsible for many relevant developments in Artificial Intelligence (AI). Successes such as beating the world champion of Go \cite{Silver2016MasteringTG} and solving numerous complex games without any human intervention \cite{Schrittwieser_2020} were relevant milestones in AI, providing optimal planning without supervision. RL is paramount in complex real-world problems such as self-driving vehicles \cite{kiran2021deep}, automated trading \cite{liu2020finRL,Mosavi_2020}, recommender systems \cite{afsar2021reinforcement}, quantum physics \cite{Dalgaard_2020}, among many others. Recent advancements in RL are strongly associated with advances in Deep Learning \cite{dl_bengio} since scaling to large state/action space environments is possible, as opposed to tabular RL \cite{sutton_RL}.\\
Previous results suggest that RL agents obeying the rules of quantum mechanics can outperform classical RL agents \cite{Dunjko2016,Dunjko2017a,Paparo2014,Sequeira_2021,Dunjko_2018,Saggio_2021}. However, these suffer from the same scaling problem as classical tabular RL: they do not scale easily to real-world problems with large state-action spaces. Additionally, the lack of fault-tolerant quantum computers \cite{Preskill1997FaulttolerantQC} further compromises the ability to handle problems of significant size.\\
Variational Quantum Circuits (VQCs) are a viable alternative since state-action pairs can be parameterized, enabling, at least in theory, a reduction in the circuit's complexity. Moreover, VQCs could enable shallow enough circuits to be confidently executed on current NISQ (\textit{Noisy Intermediate Scale Quantum}) hardware \cite{Preskill_2018} without resorting to typical brute force search over the state/action space as in the quantum tabular setting \cite{Sequeira_2021,Dunjko2016}. Variational models are also referred to as approximately universal quantum neural networks \cite{farhi2018classification, Schuld_2021}. Nevertheless, fundamental questions on the expressivity and trainability of VQCs remain to be answered, especially from a perspective relevant to RL.\\
This paper proposes an RL agent's policy resorting to a shallow VQC and studies its effectiveness when embedded in the Monte-Carlo-based policy gradient algorithm REINFORCE \cite{Williams2004SimpleSG} throughout standard benchmarking environments. However, benchmarking variational algorithms for classical environments exhibit a trade of information between a quantum and a classical channel that incurs an overhead from encoding classical information into the quantum processor. Efficient encoding of real-world data constitutes a real bottleneck for NISQ devices, with the consequence of neglecting any potential quantum advantage \cite{LaRose2020RobustDE}. In the case of a quantum agent-environment interface, the cost of data encoding can often be neglected, and there is room for potential quantum advantages from quantum data \cite{Huang_2021}. In optimal quantum control, gate fidelity is improved by exploiting the full knowledge of the system's Hamiltonian \cite{doi:10.1146/annurev-control-061520-010444}. However, such methods are only viable when the system's dynamics are known. Thus, applying variational quantum methods may indeed be relevant \cite{lamata_2021}. In this setting, we considered a quantum RL agent that optimizes the gate fidelity in a model-free setting, learning directly from the interface with the noisy environment.\\ 
\\
The main contributions of this paper are:
\begin{itemize}
\item Design of a variational softmax-policy using a shallow VQC similar to or outperforming long-term cumulative reward compared to a restricted class of classical neural networks used in a set of standard benchmarking environments and the problem of quantum state preparation, using a fraction of the number of trainable parameters.
\item Demonstration of a logarithmic sample complexity concerning the number of parameters in gradient estimation.
\item Empirical verification of different parameter initialization strategies for variational policy gradients. 
\item Study of the barren plateau phenomenon in quantum policy gradient optimization using the Fisher information matrix spectrum.
\end{itemize}

The rest of the paper is organized as follows. Section \ref{sec: Related work} reviews quantum variational RL's state-of-the-art. Section \ref{sec: Policy Gradients} summarizes the theory behind the classical policy gradient algorithm used in this work. Section \ref{sec: QPG} details each block of the proposed VQC and the associated quantum policy gradient algorithm. Section \ref{sec: gradient estimation} explores trainability under gradient-based optimization using quantum hardware and its corresponding sample complexity. Section \ref{sec: performance} presents the performance of the quantum variational algorithm in simulated benchmarking environments. Section \ref{sec: advantage} analyzes the number of parameters trained and the Fisher Information spectrum associated with the classical/quantum policy gradient. Section \ref{sec: conclusion} closes the paper with some concluding remarks and suggestions for future work.

\section{Related Work}\label{sec: Related work}

Despite numerous publications focusing on Quantum Machine Learning (QML), the literature on variational methods applied to RL remains scarce. Most results to date focus on value-based function approximation rather than policy-based. Chen et al. \cite{Chen2020} use VQCs as quantum value function approximators for discrete state spaces, and, in \cite{Lockwood2020}, the authors generalize the former result to continuous state spaces. Lockwood et al. \cite{Lockwood2021PlayingAW} show that simple VQC-inspired Q-Networks (i.e., state-action value approximators) based on Double Deep Q-Learning are not adequate for the Atari games, Pong and Breakout. Sanches et al. \cite{sanches2021short} proposed a hybrid quantum-classical policy-based algorithm to solve real-world problems like vehicle routing. In \cite{wu2021quantum}, the authors proposed a variational actor-critic agent, which is the only work so far operating on the quantum-quantum context of QML \cite{Aimeur2007}, i.e., a quantum agent acting upon a quantum environment. The authors suggest that the variational method could solve quantum control problems. Jerbi et al. \cite{jerbi2021variational} propose a novel quantum variational policy-based algorithm achieving better performance than previous value-based methods in a set of standard benchmarking environments. Their architecture consists of repeated angle-encoding to increase the expressivity of the variational model, i.e., increasing the number of functions of the input state that the model can represent \cite{Schuld_2021}. Compared with \cite{jerbi2021variational}, our work shows that a simpler variational architecture composed of a shallow ansatz, consisting of a two-qubit entangling gate and two single-qubit gates \cite{RevModPhys.94.015004} with a single encoding layer can be considered for standard benchmarking environments. Variational policies can be devised with decreased depth and fewer trainable parameters. The type of functions our circuit can represent is substantially smaller when compared to \cite{jerbi2021variational}. However, simpler classes of policies may be beneficial in the language of generalization and overfitting. Furthermore, compared to \cite{jerbi2021variational}, this work considers a more trivial set of observables for the measurement of the quantum circuit, leading to fewer shots necessary to estimate the agent's policy and respective policy gradient.

\section{Policy Gradients}\label{sec: Policy Gradients}

Policy Gradient methods try to learn a parameterized policy $\pi(a\lvert  s,\theta) = \mathbb{P}\{ a_t = a \lvert  s_t = s , \theta_t = \theta \}$, where $\theta \in \mathbb{R}^k$ is the parameter vector of size $k$, $s$ and $a$ are the state and action, respectively, and $t$ is the time instant, that can optimally select actions without resorting to a value function. These methods try to maximize a performance measure $J(\theta)$, performing gradient \textit{ascent} on $J(\theta)$

\begin{equation}
\theta_{i+1} = \theta_i + \eta \nabla_{\theta_i} J(\theta_i)
\end{equation}
\noindent
where $\eta$ is the learning rate. Provided that the action space is discrete and relatively small, then the most prominent way of balancing exploration and exploitation is by sampling an action from a \textit{Softmax-Policy}, also known as Neural Policy \cite{Agarwal2019ReinforcementLT}: 
\begin{equation}
\pi(a \lvert  s,\theta) = \frac{e^{h(s,a,\theta)}}{\sum_{b \in A} e^{h(s,b,\theta)}}
\label{eq: softmax-policy}
\end{equation}
\noindent
where $h(s,a,\theta) \in \mathbb{R}$ is a numerical preference for each state-action pair and $A$ is the action set. For legibility, $A$ will be omitted whenever a policy similar to equation \eqref{eq: softmax-policy} is presented. The policy gradient theorem \cite{Sutton1999PolicyGM} states that the gradient of the objective function can be written as a function of the policy itself. In general, the Monte-Carlo policy gradient known as REINFORCE \cite{Williams2004SimpleSG}, computes the gradient of samples obtained from $N$ trajectories of length $T$, also known as the \textit{horizon} under the parameterized policy, as in Equation \eqref{eq: policy gradient estimator}.
        
\begin{equation}
\nabla_{\theta} J(\theta) = \frac{1}{N} \sum_{i=0}^{N-1}\sum_{t=0}^{T-1} G_t(\tau_i) \nabla_{\theta} \log \pi(a_{t_i}\lvert  s_{t_i},\theta)
\label{eq: policy gradient estimator} 
\end{equation}
\noindent
where $G_t(\tau)$ is the $\gamma$-discounted cumulative reward per time step, known as the \textit{return} (see Equation \eqref{eq: return_t}) derived from trajectory's return $G(\tau)$ (see Equation \eqref{eq: return}).

\noindent\begin{minipage}{.5\linewidth}
\begin{equation}
 G(\tau) = \sum_{t=0}^{T-1} \gamma^t r_{t+1}
 \label{eq: return}
\end{equation}
\end{minipage}%
\begin{minipage}{.5\linewidth}
\begin{equation}
 G_t(\tau) = \sum_{t' = 0}^{T-t-1} \gamma^{t'} r_{t'+t}
 \label{eq: return_t}
\end{equation}
\end{minipage}
\noindent
A known limitation of the REINFORCE algorithm is due to Monte Carlo estimates. Stochastically sampling the trajectories results in gradient estimators with high variance, which deteriorate the performance as the environment's complexity increases \cite{variance_reduction}. The REINFORCE estimator can be improved by leveraging a control variate known as \textit{baseline} $b(s_t)$, without increasing the number of samples $N$. Baselines are subtracted from the return such that the optimization landscape becomes smooth. The REINFORCE with baseline gradient estimator is represented in Equation \eqref{eq: policy gradient baseline}, and the complete algorithm is presented in Algorithm \ref{alg: reinforce}. 

\begin{equation}
      \nabla_{\theta} J(\theta) = \frac{1}{N} \sum_{i=0}^{N-1}\sum_{t=0}^{T-1} (G_t(\tau_i) - b(s_{t_i})) \nabla_{\theta} \log \pi(a_{t_i} \lvert  s_{t_i} , \theta)
    \label{eq: policy gradient baseline}
\end{equation}
\noindent
For the benchmarking environments in Section \ref{sec: performance}, the average return was used as a baseline, calculated as in equation \eqref{eq: baseline}.

    \begin{equation} 
    b(s_t) = {1 \over {N}} \sum_{i=0}^{N-1} G_t(\tau_i)
    \label{eq: baseline}
    \end{equation}

\begin{algorithm}
\caption{REINFORCE with baseline}\label{alg: reinforce}
\begin{algorithmic}
\Require $\theta \in \mathbb{R}^k$ , learning rate $\eta$, horizon $T$

\While{True}
\For{$i=0 \ldots N-1$}

     \State{Following $\pi_{\theta}$, generate trajectory of the form 
     $$\tau_i = \{(s_0,a_0,r_0),\dots ,(s_{T-1},a_{T-1},r_{T-1})\}$$} 
\EndFor

\State{Compute gradient with baseline as in Equation \eqref{eq: policy gradient baseline}}
\State{update parameters via gradient ascent $\theta = \theta + \eta \nabla_{\theta} J(\theta)$}

\EndWhile
\end{algorithmic}
\end{algorithm}



\section{Quantum Policy Gradients}
\label{sec: QPG}

This section details the proposed VQC-based policy gradient. Numerical preferences $h(s, a,\theta) \in \mathbb{R}$ are the output of measurements in a given parameterized quantum circuit. The result can be represented as the expectation value of a given observable or the probability of measuring a basis state. We resort to the former since it allows for more compact representations of objective functions \cite{bharti2021noisy}. Additionally, the type of ansatz used by the proposed VQC implies that $\theta \in \mathbb{R}^k$ is a high dimensional vector corresponding to the angles of arbitrary single-qubit rotations.\\
VQCs are composed of four main building blocks, as represented in \autoref{fig: vqc building blocks}. Initially, a state preparation routine or \textit{embedding}, $S$, encodes data points into the quantum system. Next, a unitary $U(\theta)$ maps the data into higher dimensions of the Hilbert space. Such a parameterized model corresponds to linear methods in quantum feature spaces. Expectation values returned from a measurement scheme are finally post-processed into the quantum neural policy. A careful analysis of each block of Figure \ref{fig: vqc building blocks} follows. Moreover, the sample complexity of estimating the quantum policy gradient is analyzed in Section \ref{sec: gradient estimation}.

\begin{figure}[h]
 \centering
 \includegraphics[width=0.5\textwidth]{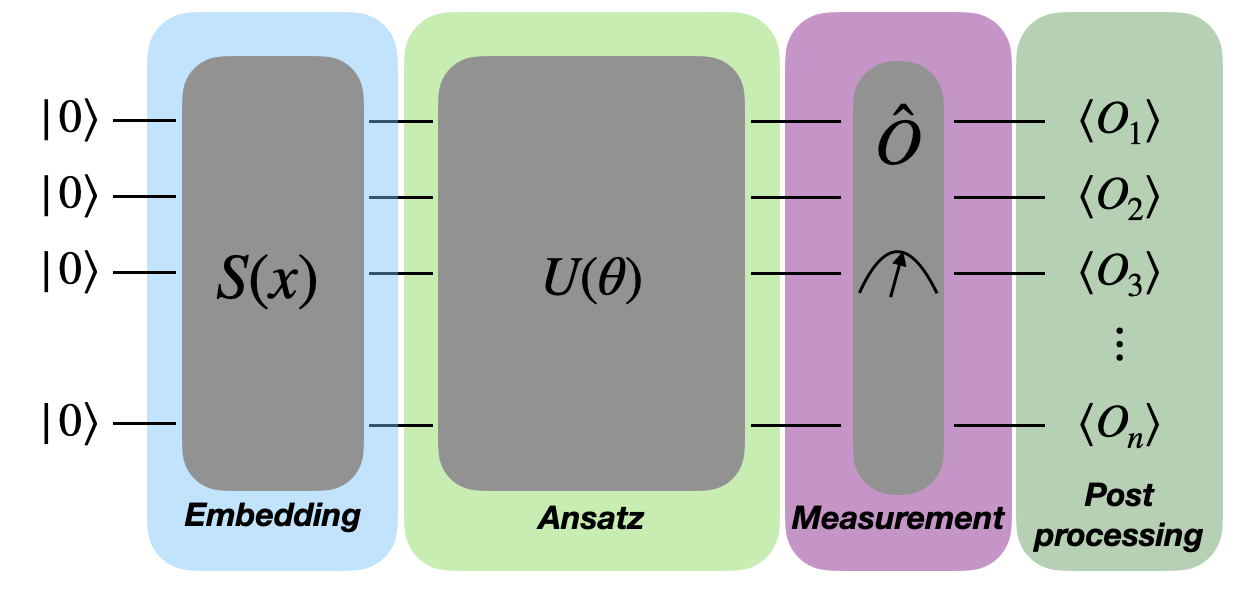}
 \caption{Building blocks of Variational Quantum Circuits.}
 \label{fig: vqc building blocks}
\end{figure}

\subsection{Embedding}\label{subsec: embedding}
Unlike classical algorithms, the state-preparation routine is a crucial step for any variational quantum algorithm. There are numerous ways of encoding classical data into a quantum processor \cite{schuld_ml_book}.
%In kernel methods, a careful choice of the encoding strategy must be made, given that it is the most crucial step towards achieving advantage \cite{Schuld2021QuantumML}. Therefore, there is a need to design quantum encodings that are classically hard to simulate, like the IQP encoding \cite{Havl_ek_2019}. This work uses VQCs to derive Neural Policies for RL agents. We can relax about the hardness of the encoding itself, given that empirical advantage depends intrinsically on the VQC as a whole. 
Angle encoding \cite{LaRose2020RobustDE} is used to allow continuous-state spaces. Arbitrary Pauli rotations $\sigma \in \{\sigma_x , \sigma_y , \sigma_z\}$ can encode a single feature per qubit. Hereby, given an agent's state $s$ with $n$ features, $s = \{s_0, s_2, \dots s_{n-1}\}$, $\sigma_x$ rotations are used, requiring $n$ qubits to encode $\ket{s}$, as indicated by Equation \eqref{eq: angle_enc}.

\begin{equation}
 \ket{s} = \bigotimes_{i=0}^{n-1} e^{-j \sigma_x s_i} \ket{b_i}
 \label{eq: angle_enc}
\end{equation}
\noindent
 where $\ket{b_i}$ refers to the $i^{\text{th}}$ qubit of an $n$-qubit register initially in state $\ket{0^n}$(represented w.l.g as $\ket{0}$ from now on). Each feature needs to be normalized such that $s_i \in [-\pi,\pi]$. Since the range of each feature is usually unknown, this work resorts to normalization based on the $L_{\infty}$ norm. The main advantage of angle encoding lies in the simplicity of generating the encoding, given the composition of solely $n$ single-qubit gates, thus giving rise to a circuit of depth 1. In contrast, the main disadvantage is the linear dependence between the number of qubits and the number of features characterizing the agent's state and the poor representational power, at least in principle \cite{Schuld2021QuantumML}.

\subsection{Parameterized model}\label{subsec: ansatz}

To the best of the authors' knowledge, no \textit{problem-inspired} ansatz exploiting the physics behind the problem is known in RL applications. This can be explained by the difficulty of expressing and training RL agent's policies as Hamiltonian-based evolution models \cite{bharti2021noisy}. Moreover, since the goal is to design a NISQ ansatz to capture the agent's optimal policy in different environments, this work uses a parameterized model from the family commonly referred to as \textit{hardware-efficient} ansatze \cite{bharti2021noisy}. Such models behave similarly to a classical feed-forward neural network. The main advantage of this family of ansatze is its versatility, accommodating encoding symmetries and bringing correlated qubits closer for depth reduction \cite{Cerezo_Nature}. The ansatz consists of an alternating-layered architecture composed of single-qubit gates followed by a cascade of entangling gates as pictured in \autoref{fig: ansatz}.

\begin{figure}[h]
 \centering
 \includegraphics[width=0.6\textwidth]{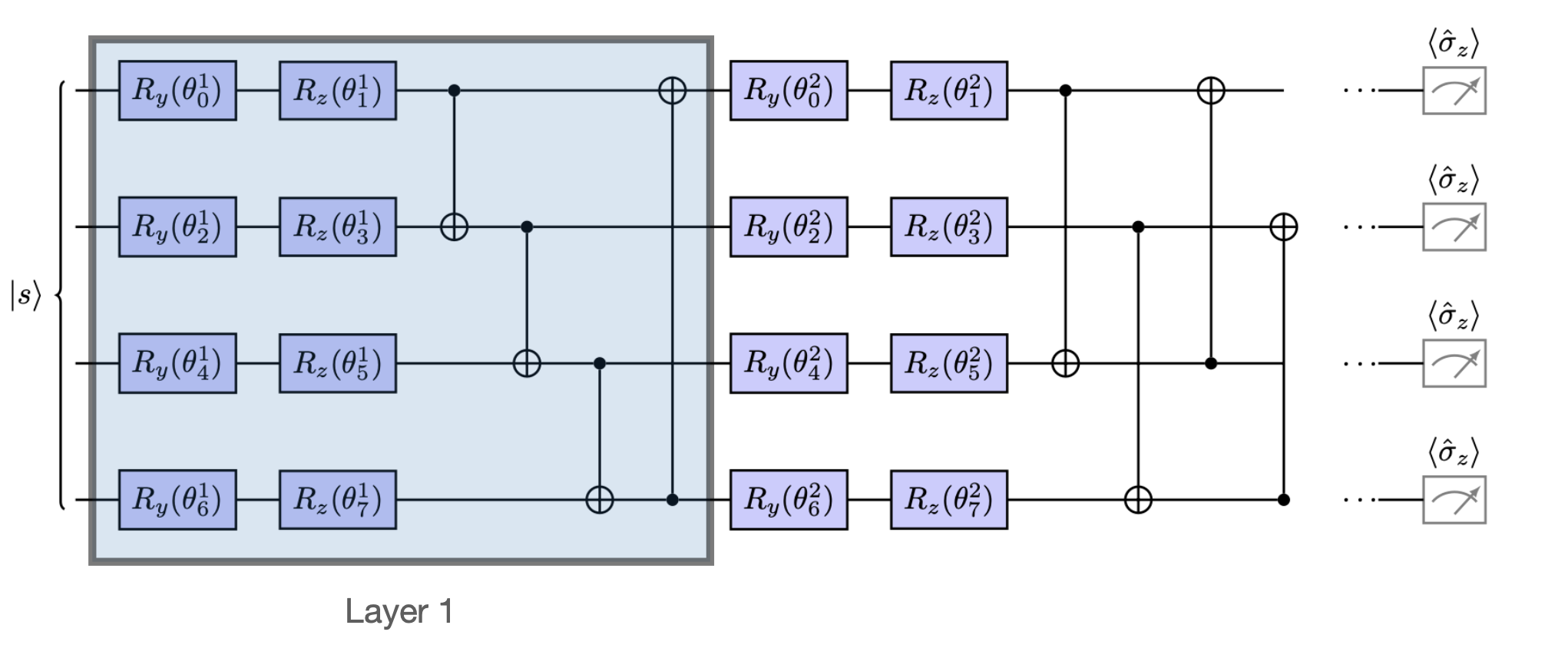}
 \caption{Hardware-efficient ansatz for RL based on single-qubit $R_y, R_z$ rotation gates.}
 \label{fig: ansatz}
\end{figure}
\noindent
A single layer is composed of two single-qubit $\sigma_y, \sigma_z$ rotation gates per qubit, followed by a cascade of entangling gates, such that features are correlated in a highly entangled state. The ansatz includes $2n$ single-qubit rotation gates per layer, each gate parameterized by a given angle. Therefore, there are $2nL$ trainable parameters for $L$ layers. The entangling gates follow a pattern that changes over the number of layers, inspired by the circuit-centric classifier design \cite{Schuld_2021}. The pattern follows a modular arithmetic $\mbox{CNOT}[i, (i+l) \mod n]$ where $i \in [1, \dots ,n]$ and $l \in [1, \dots ,L]$ indexes the layers. Increasing the number of layers increases the correlation between features and expressivity.

\subsection{Measurement}\label{sec: measurements}

An arbitrary state $\ket{\psi} \in \mathbb{C}^{2^n}$ is represented by an arbitrary superposition over the basis states, as in Equation \eqref{eq: basis-states}.

\begin{equation}
    \ket{\psi} = \sum_{i=0}^{2^{n-1}} c_i \ket{\psi_i}
    \label{eq: basis-states}
\end{equation}
\noindent
Measuring the state $\ket{\psi}$ in the computational basis ($\sigma_z$ basis) collapses the superposition into one of the basis states $\ket{\psi_i}$ with probability $\lvert c_i \rvert^2$, as given by the Born rule \cite{chuang_book}. In general, the expectation value of some observable $\hat{O}$, is given by the summation of each possible outcome, i.e., the eigenvalue $\lambda_i$ weighted by its respective probability $p_i = \lvert c_i\rvert^2$ as in Equation \eqref{eq: expectation-observable}.

\begin{equation}
    \langle \hat{O} \rangle = \expval{\hat{O}}{\psi} =  \sum_{i=0}^{2^{n-1}} \lambda_i p_i
    \label{eq: expectation-observable}
\end{equation}
\noindent
Let $\hat{O}$ be the single-qubit $\sigma_z^i$ measurement, applied to the $i^{\text{th}}$-qubit. Given that the $\sigma_z$ eigenvalues are $\{-1,1\}$, the expectation value $\langle \sigma_z^i \rangle$ can be obtained by the probability $p_0$ of the qubit being in the state $\ket{0}$ as $\langle \sigma_z^i \rangle = 2p_0 - 1$. Notice that in practice, $p_0$ needs to be estimated from several circuit repetitions to obtain an accurate estimate of the expectation value.
\noindent
Let the state $\ket{\psi}$ be the quantum state obtained from the encoding of an agent's state via $S(s)$, and the parameterized block $U(\theta)$, as in Sections \ref{subsec: embedding} and \ref{subsec: ansatz} respectively. Let $\langle \sigma_z^i \rangle$ be the quantum analogue of the numerical preference for action $i$, which we represent by $\langle a_i \rangle$ for clarity. Its expectation can be formally described by Equation \eqref{eq: expectation}.

\begin{equation}
 \langle a_i \rangle_{\theta} =  \expval{S(s)^{\dagger}U(\theta)^{\dagger} \sigma_z^{i} U(\theta)S(s)}{0}
 \label{eq: expectation}
\end{equation}
\noindent
For a policy with $\lvert A\rvert$ possible actions, each $\sigma_z$ measurement corresponds to the numerical preference of each action. Thus, $\lvert A\rvert$ single-qubit estimated expectation values are needed. If the number of features in the agent's state is larger than the number of actions, the single-qubit measurements occur only on a subset of qubits. Such measurement scheme is qubit-efficient \cite{schuld_ml_book}. 
Figure \ref{fig: vqc_4_measures} represents the full VQC for an environment with four feature states and four actions with three parameterized layers.
\begin{figure}[h]
 \centering
 \includegraphics[width=\textwidth]{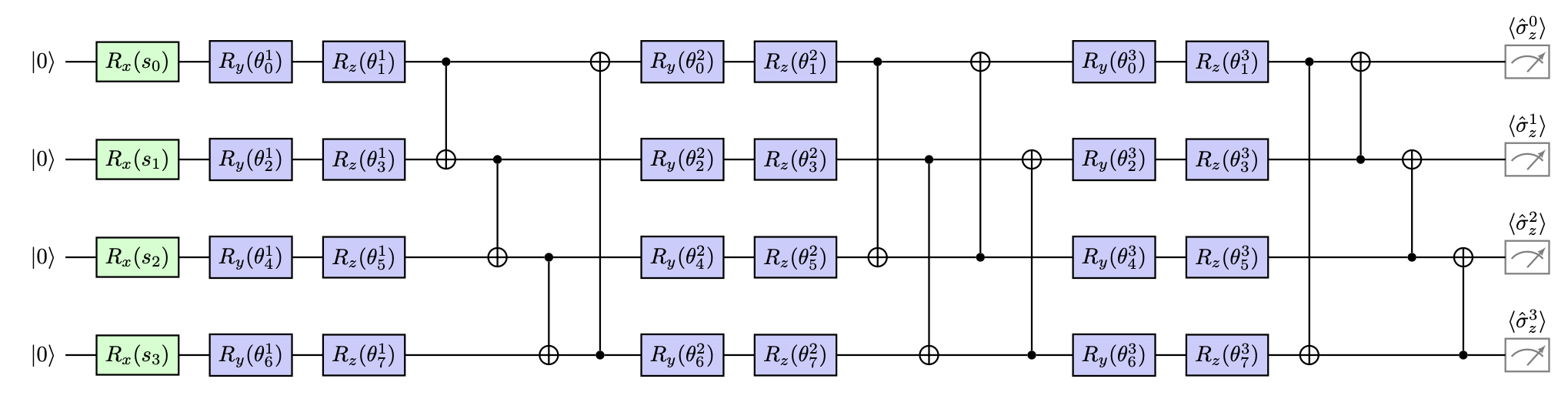}
 \caption{Variational Quantum Circuit for Policy-based RL with three parameterized layers.}
 \label{fig: vqc_4_measures}
\end{figure}

\subsection{Classical Post-processing}\label{sec: post_process}
Measurement outcomes representing numerical preferences $h(s, a,\theta)=\langle a \rangle_{\theta}$ are classically post-processed to convert the estimated expectation values to the final quantum neural policy, as given by Equation \eqref{eq: quantum_softmax_policy}. 

\begin{equation}
 \pi(a\mid s,\theta) = {{e^{\langle a \rangle_{\theta}}} \over {\sum_{b} e^{\langle b \rangle_{\theta}}}}
 \label{eq: quantum_softmax_policy}
\end{equation}
\noindent
Equation \eqref{eq: quantum_softmax_policy} imposes an upper bound on the greediness of $\pi$. It will always allow for exploratory behavior, which can negatively impact the performance of RL agents, especially in deterministic environments. As an example, consider a 2-action environment with

\begin{equation*}
 \pi = [\pi(a_0\mid s,\theta) , \pi(a_1\mid s,\theta)]
 \label{eq: pi}
\end{equation*}
The entries of $\pi$ are given by Equation \eqref{eq: quantum_softmax_policy} and the actions' estimated expectation values $[\langle a_0 \rangle_{\theta}, \langle a_1 \rangle_{\theta}]$. As these are bounded as $\langle \sigma_z \rangle \in [-1,1]$, the maximum difference between action preferences occurs when the estimated vector is $[\langle a_0 \rangle_{\theta} = -1, \langle a_1 \rangle_{\theta} = 1]$. The corresponding softmax normalized vector is:

\begin{equation*}
 \pi_a = [\pi(a_0\mid s,\theta) , \pi(a_1\mid s,\theta)] = [0.88,0.12]
 %\label{eq: 2-policy}
\end{equation*}
\noindent
In this case, the policy always has a $\sim 0.1$ probability of selecting the worst action; the same \textit{rationale} applies to larger action sets. Thus, a trainable parameter $\beta$ is added to the quantum neural policy as in Equation \eqref{eq: beta-policy}:

\begin{equation}
 \pi(a\lvert s,\theta) = {{e^{\beta\langle a \rangle_{\theta}}} \over {\sum_{b} e^{\beta\langle b \rangle_{\theta}}}} 
 \label{eq: beta-policy}
\end{equation}
$\beta$ has the effect of scaling the output values from the quantum circuit measurements, resembling an energy-based model. Instead of decreasing $\beta$ over time, we treat it as a hyperparameter to be tuned along with $\theta$. The optimization sets $\beta$, assuring convergence towards the optimal policy.

\subsection{Gradient Estimation}\label{sec: gradient estimation}

This section develops upper bounds on both the number of samples and the number of circuit evaluations necessary to obtain an $\epsilon$-approximation of the policy gradient, as given by Equation \eqref{eq: policy gradient estimator}, restated here for completion:
\begin{equation*}
\nabla_{\theta} J(\theta) = \frac{1}{N} \sum_{i=0}^{N-1}\sum_{t=0}^{T-1} G_t(\tau_i) \nabla_{\theta} \log \pi(a_{t_i}\lvert  s_{t_i},\theta)
\end{equation*}
\noindent
The gradient $\nabla_{\theta} J(\theta)$ can be estimated using the same quantum device that computes expectations $\langle a_i \rangle_{\theta}$, via parameter-shift rules \cite{Schuld2019EvaluatingAG}. These rules require the policy gradient to be framed as a function of gradients of observables, as given by Equation \eqref{eq: log_policy_expansion-LPS}.

\begin{equation}
 \nabla_{\theta} \log \pi(a\lvert s,\theta) \;=\; \beta \left(\nabla_{\theta} \langle a \rangle_{\theta} - {\sum_b \pi(b\lvert s,\theta) \nabla_{\theta} \langle b \rangle_{\theta}}\right)
 \label{eq: log_policy_expansion-LPS}
\end{equation}
\noindent
By combining equations \eqref{eq: policy gradient estimator} and \eqref{eq: log_policy_expansion-LPS}, the quantum policy gradient estimator is given by Equation \eqref{eq: quantum policy gradient-LPS}: 
 
 \begin{equation}
     \nabla_{\theta} J(\theta) \;=\; \frac{1}{N} \sum_{i=0}^{N-1}\sum_{t=0}^{T-1} G_t(\tau_i) \beta \left(\nabla_{\theta} \langle a_{t_i} \rangle_{\theta} - {\sum_{b_{t_i}} \pi(b_{t_i}\mid s_{t_i},\theta) \nabla_{\theta} \langle b_{t_i} \rangle_{\theta}}\right)
 \label{eq: quantum policy gradient-LPS}
 \end{equation}
\noindent
The number of samples associated with Equation \eqref{eq: quantum policy gradient-LPS} is defined as the number of visited states. Since there are $N$ trajectories (sequences of actions, $\tau_i$), each visiting $T$ states, the total number of samples is $\mathcal{O}(NT)$.
\noindent
Lemma \ref{lemma: sample complexity-LPS} provides an upper bound for $N$ such that the policy gradient is $\epsilon_{\nabla}$-approximated with probability $1-\delta_{\nabla}$. 
 
 %\begin{restatable}[$\epsilon$-approximation of the policy-gradient]{lemma}{sample}
 %\label{lemma: sample complexity}
 %Let $\theta \in \mathbb{R}^k$, $R_{max}$ be maximum reward achievable in a single time step, $T$ the horizon, and %$\nabla_{\theta} J(\theta)$ the expected policy gradient. To produce, with probability $1-\delta$, an $\epsilon$-approximated %estimate of the policy gradient, $\hat{\nabla}_{\theta} J(\theta)$

 %\begin{equation}
 %\lvert  \hat{\nabla}_{\theta} J(\theta) - \nabla_{\theta} J(\theta) \rvert  \;\leq\; \epsilon
 %\end{equation}
 %\noindent
 %the quantum policy gradient requires a number of samples given by
 %
 %\begin{equation}
 %\mathcal{O}\left( \frac{8\beta^2 R_{max}^2 T^3}{\epsilon^2 (\gamma - 1)^4} log ({2k \over \delta}) \right)
 %\end{equation}
%\end{restatable}
%\noindent
\begin{restatable}[$\epsilon_{\nabla}$-approximation of the policy-gradient]{lemma}{sample}
 \label{lemma: sample complexity-LPS}
 Let $\theta \in \mathbb{R}^k$, $k$ being the number of parameters, $R_{max}$ be the maximum possible reward in any time step, $T$ the horizon, and $\nabla_{\theta} J(\theta)$ the expected policy gradient. The policy gradient, $\hat{\nabla}_{\theta} J(\theta)$, can be $\epsilon_{\nabla}$-approximated, with probability $1-\delta_{\nabla}$

 \begin{equation}
 \lvert  \hat{\nabla}_{\theta} J(\theta) - \nabla_{\theta} J(\theta) \rvert  \;\leq\; \epsilon_{\nabla}
 \end{equation}
 \noindent
 using a number of samples given by
 
 \begin{equation}
 NT \approx \mathcal{O}\left( \frac{8\beta^2 R_{max}^2 T^3}{\epsilon_{\nabla}^2 (\gamma - 1)^4} \log ({2k \over \delta_{\nabla}}) \right)
 \end{equation}
 
\end{restatable}
\noindent
The most relevant insight drawn from Lemma \ref{lemma: sample complexity-LPS} is that it establishes that for obtaining an $\epsilon_{\nabla}$-approximated policy gradient, the algorithm needs a number of samples that grows logarithmically with the total number of parameters. The proof of Lemma \ref{lemma: sample complexity-LPS} is presented in detail in Appendix \ref{appendix: gradient_estimation_LPS}.
\\
Gradient-based optimization can be performed using the same quantum device that computes expectations $\langle a_i \rangle_{\theta}$, via parameter-shift rules \cite{Sweke_2020,Schuld2019EvaluatingAG}, which compute the gradient of an observable w.r.t a single variational parameter concerning rotation angles of quantum gates. Parameter-shift rules are given by Equation \eqref{eq: parameter-shift rules-LPS}:
\begin{equation}
\nabla_{\theta_i} \langle a_i \rangle_{\theta} = \frac{1}{2} \left[ \langle a_i \rangle_{\theta + \frac{\pi}{2}} - \langle a_i \rangle_{\theta - \frac{\pi}{2}} \right]
\label{eq: parameter-shift rules-LPS}
\end{equation}
\noindent
The gradient's accuracy depends on the expectation values, $\langle a \rangle_{\theta}$. These are estimated for each sample and action using several repetitions of the quantum circuit or shots. Lemma \ref{lemma: query complexity-LPS} establishes an upper bound on the total number of shots required to reach an $\epsilon_{\langle \rangle}$-approximated policy gradient, with probability $1-\delta_{\langle \rangle}$.
 
  \begin{restatable}[Total number of quantum circuit evaluations]{lemma}{query}
 \label{lemma: query complexity-LPS}
 Let $\theta \in \mathbb{R}^k$, $\mathcal{O}(NT)$ be the sample complexity given by Lemma \ref{lemma: sample complexity-LPS}, and $\lvert A \rvert$ the number of available actions. With probability $1-\delta_{\langle \rangle}$ and approximation error $\epsilon_{\langle \rangle}$, the quantum policy gradient algorithm requires a number of shots given by
 
 \begin{equation}
 \mathcal{O}\left( \frac{\lvert A \rvert NT}{\epsilon_{\langle \rangle}^2}  \log (\frac{2k}{\delta_{\langle \rangle}})\right)
 \end{equation}
\end{restatable}

 \noindent
 Similarly to Lemma \ref{lemma: sample complexity-LPS}, it is shown that the accuracy of the policy gradient, as a function on the total number of shots, grows logarithmically with the total number of parameters. The proof of Lemma \ref{lemma: query complexity-LPS} is presented in detail in Appendix \ref{appendix: query_complexity_LPS}.

\section{Performance in simulated environments}\label{sec: performance}

This section examines the performance of the proposed quantum policy gradient through standard benchmarking environments from the OpenAI Gym library \cite{brockman2016openai}. Moreover, the quantum policy gradient was also tested in a handcrafted quantum control environment. In this setting, a quantum agent was designed to learn to prepare the state $\ket{1}$ with high fidelity, starting from the ground state $\ket{0}$. The empirical reward over the number of episodes was used to discern the performance of both classical and quantum models. The best-performing classical neural network was selected from a restrictive set of networks composed of at most two hidden linear layers. All quantum circuits were built using the Pennylane library \cite{bergholm2020pennylane} and trained using the PyTorch automatic differentiation backend \cite{paszke2017automatic} to be directly compared with classical models built with the same library. All training instances used the most common classical optimizer, ADAM \cite{kingma2017adam}.

\subsection{Numerical Experiments}\label{sec: benchmarking}
		The CartPole-v0 and Acrobot-v1 environments were selected as classic benchmarks. They have a continuous state space with a relatively small feature space (2 to 6 features) and discrete action space (2 to 3 possible actions). The reward function is similar to each environment. In Cartpole, the agent receives a reward of $+1$ at every time step. The more time the agent keeps the pole from falling, the more reward it gets. In Acrobot, the agent receives a $-1$ reward at every time step and reward $0$ once it gets to the goal state. Thus, Acrobot will be harder to master since, for the Cartpole, every action has an immediate effect as opposed to Acrobot.\\
		In the quantum control environment of state preparation, which we refer to as QControl on this point onward, for simplicity, the mapping $\ket{0} \mapsto \ket{1}$ can be characterized by a time-dependent Hamiltonian $H(t)$ of the form of Equation \eqref{eq: time_dep_hamiltonian} describing the quantum environment as in \cite{Zhang_2019}.

 \begin{equation}
 H(t) = 4J(t)\sigma_z + h\sigma_x
 \label{eq: time_dep_hamiltonian}
 \end{equation}
\noindent
Where $h$ represents the single-qubit energy gap between tunable control fields, considered a constant energy unit. $J(t)$ represents the dynamical pulses controlled by the RL quantum agent in a model-free setting. The learning procedure defines a fixed number of steps $N = 10$, from which the RL agent must be able to create the desired quantum state. The quantum environment prepares the state associated with the time step $t+1$, given the gate-based Hamiltonian at time step $t$, $U(t)$:
\begin{equation}
 \ket{\psi_{t+1}} = U(t)\ket{\psi}
 \label{eq: evolution}
 \end{equation}
\noindent
The reward function is naturally represented as the fidelity between the target state $\ket{\psi_T} = \ket{1}$ and the prepared state $\ket{\psi_t}$ naturally serves as the reward $r_t$ for the agent at time step $t$, as in Equation \eqref{eq: fidelity}.
 
 \begin{equation}
 r_t = {\lvert \langle \psi_{t} \lvert  \psi_{T} \rangle \rvert }^2
 \label{eq: fidelity}
 \end{equation}
\noindent
 Using the policy gradient algorithm of Section \ref{sec: QPG}, the goal is to learn how to maximize fidelity. Figure \ref{fig: q_control} depicts the agent-environment interface.

 \begin{figure}[h]
 \centering
 \includegraphics[width=0.6\textwidth]{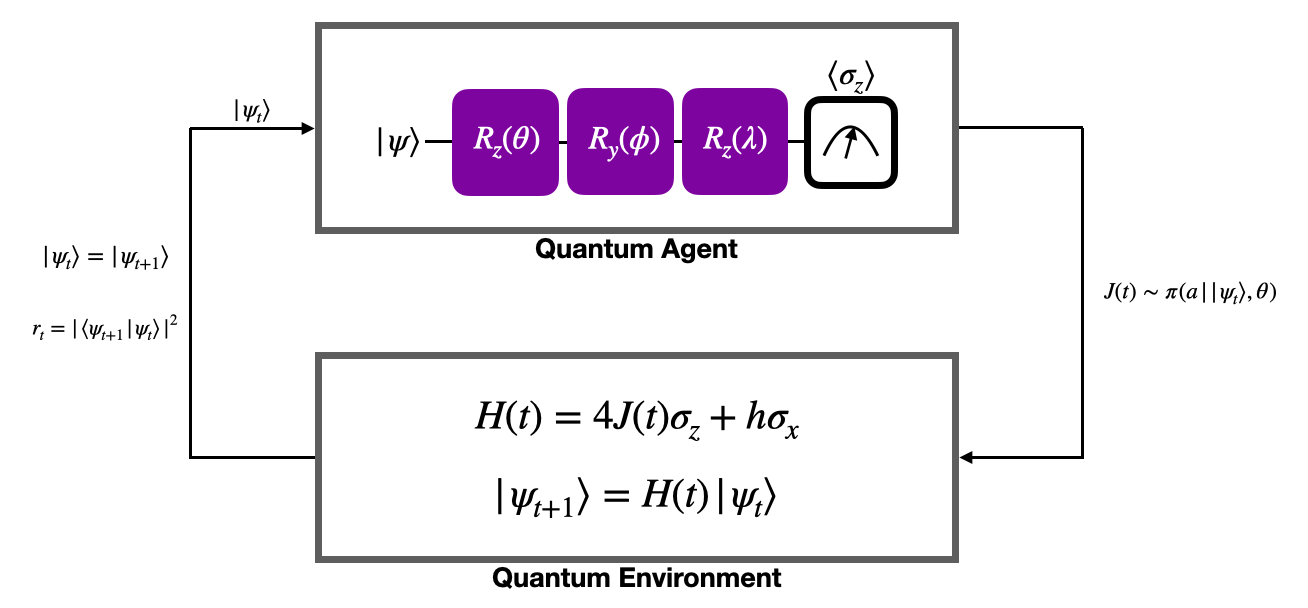}
 \caption{Agent-Environment interface for quantum control.}
 \label{fig: q_control}
 \end{figure}
\noindent
 Each sequence of $N$ pulses corresponds to an episode. The quantum agent should learn the optimal pulse sequence that maps to the state with maximum fidelity as the number of episodes increases. The quantum variational architecture selected was the same as described in Section \ref{sec: QPG}. In this setting, the main difference is the lack of encoding. The quantum agent receives the quantum state from the corresponding time-step Hamiltonian applied at each time step. However, since the environment is simulated, the qubit is prepared in the state of time step $t$ and then fed to the variational quantum policy. In this setting, it is considered the binary action-space $A = [0,1]$(apply pulse $A=1$ or not, $A=0$). A sequence of $N$ actions corresponds to $N$ pulses. A performance comparison is made relative to classical policy gradients. In this case, the corresponding state vector associated with the qubit was explicitly encoded at each time step, considering both real and imaginary components. 
All environment specifications are presented in Table \ref{tab: envs}.

\begin{table}[!ht]
\centering
\vspace{10pt}
%\begin{tabular}{\textwidth}{\mid X\mid X\mid X\mid X\mid X\mid X\mid }
\begin{tabular}{cccccc}
\hline
\multirow{2}{*}{\bf Environment} & \multirow{2}{*}{\bf \#F}     & \multirow{2}{*}{\bf \#A}     & {\bf Reward}  & {\bf Max \#s}  & {\bf Terminal}    \\
 &      & & (per step)     & (per episode) & {\bf states}   \\
\hline
\multirow{3}{*}{CartPole-v0} & \multirow{3}{*}{4} & \multirow{3}{*}{2} & \multirow{3}{*}{1}  & \multirow{3}{*}{200} & Out of bounds \\
 & & & & & or reward 200 \\
 & & & & & or below horizontal line \\ \hline
Acrobot-v1 & 6 & 3 & -1  & 500 & 500 steps\\ \hline
QControl & 4 & 2 & $\lvert \bra{\psi_t} \ket{\psi} \rvert^2$ & 10 & $\lvert \bra{\psi_t} \ket{\psi} \rvert^2 \leq 10^{-4}$ \\& & & & & or 10 steps\\ \hline
%\multirow{2}{*}{MountainCar-v0} & \multirow{2}{*}{2} & \multirow{2}{*}{3} & \multirow{2}{*}{-1 + height} & \multirow{2}{*}{200} & Top of hill \\
%& & & & & or out of bounds \\ \hline
\end{tabular}
\vspace{10pt}
\caption{Description of the environments ({\footnotesize \#F: number of features; \#A: number of actions; Max \#s: maximum steps}).}
\vspace{-3pt}
\label{tab: envs}
\end{table}

 %\begin{figure}[h]
 %\centering
 %\includegraphics[width=0.5\textwidth]{images/config_2.png}
 %\caption{Description of the environments.}
 %\label{fig: envs}
 %\end{figure} \todo{Transforma esta figura numa tabela}
 
%Since the classical case consists of a neural network function approximator, the first step is to discover the best-performing classical neural network that solves each environment empirically. \\
 \noindent
 Several neural network architectures were tested for the CartPole-v0 and Acrobot-v1 environments. However, the structure is the same. Every neural network is composed of fully connected layers using a rectified linear unit (ReLU) activation function in every neuron. The output layer is the only layer that does not have ReLU activation. The depth, the total number of trainable parameters, and the existence of dropout differs from network to network. All the networks using dropout have a probability equal to 0.2. Every network was trained with an ADAM optimizer with an experimentally fine-tuned learning rate of 0.01. Figures \ref{fig: nn-init-experiments}(a) and \ref{fig: nn-init-experiments}(b) illustrate the average reward for different classical network configurations for the benchmarking environments. The results show that a fully connected neural network with a single layer of 128 and 32 neurons performs reasonably better than similar architectures for the CartPole-v0 and Acrobot-v1 environments, respectively.\\
 In the QControl environment, eight different neural networks were tested with a single hidden layer. Since the optimal neural network for this problem is still an open question, to the best of the author's knowledge, it was decided to successively increase the size of the network until it solves the task of comparing the minimum viable network with the VQC. For this set of classical architectures, the neural network with a single layer of 16 neurons was chosen since it achieves the best average fidelity as the minimum viable network solving the problem, as illustrated in Figure \ref{fig: nn-init-experiments}(c). 

\begin{figure}[!htb]
    \centering
    \includegraphics[width=\textwidth]{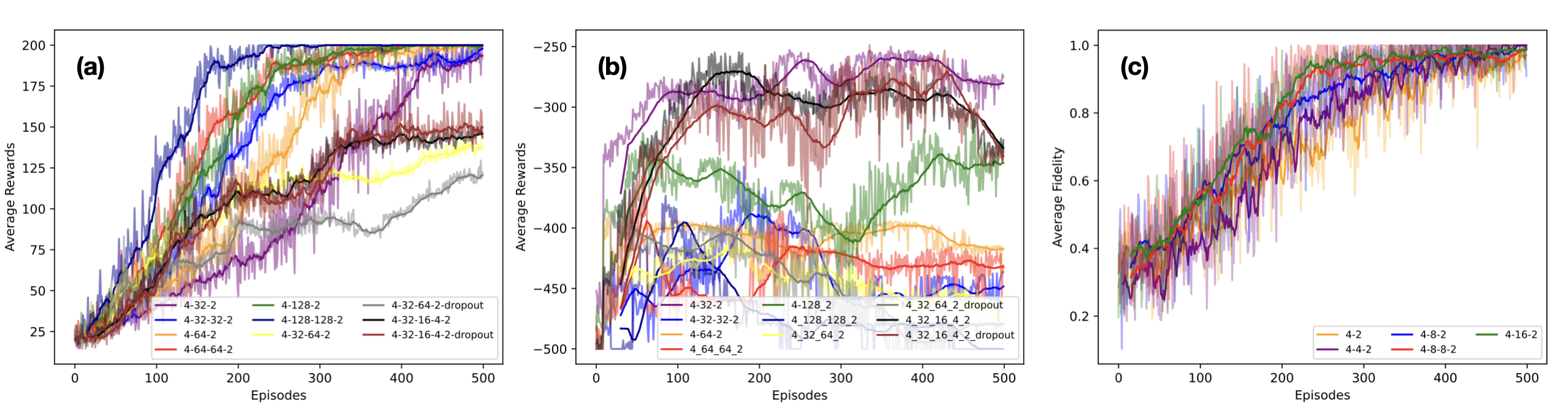}
    \caption{Different classical neural network architectures used in the three simulated environments. Panels (a), (b), and (c) represent different architectures for the Cartpole, Acrobot, and QControl environments, respectively. Each label indicates the respective network structure and if it uses dropout. Each label represents the total number of neurons in each input, hidden, and output layer. E.g., $4-4-4$ has input, hidden, and output layers with four neurons each.}
    \label{fig: nn-init-experiments}
\end{figure}
\noindent
 The second step compares the performance of the quantum neural policy of Section \ref{sec: QPG} against the aforementioned classical architecture. Increasing the number of layers in the parameterized quantum model would perhaps increase the expressivity of the model \cite{Schuld2021QuantumML}. At the same time, increasing the number of layers leads to more complex optimization tasks, given that more parameters need to be optimized. For some variational architectures, there is a threshold for expressivity in terms of the number of layers \cite{Sim_2019}. We encountered precisely this in practice. For Cartpole, the expressivity of the quantum neural policy saturates after three layers, and for the Acrobot, after four layers. From there on, the agent's performance deteriorated rather than improved.
 For the QControl environment, the classical NN was compared with a simplified version of the variational softmax policy. In this case, it was considered a VQC with the most general gate with three parameters that can approximately prepare every single-qubit state. The observables for the numerical action preference are the opposite sign computational basis measurement, i.e., $[\langle \sigma_z \rangle, -\langle \sigma_z \rangle]$. In every environment, the model's learning rate was fine-tuned by trial and error as opposed to $\beta$, which was randomly initialized. The optimal configuration for the learning rate, number of layers, and batch size used to compare are presented in table \ref{tab: config}.
 
\begin{table}[ht]
\centering
%\begin{tabularx}{\textwidth}{\mid c\mid c \mid c \mid c\mid c\mid }
\begin{tabular}{ccccc}
\hline
\textbf{Environment} & \textbf{Policy} & \textbf{Learning rate} & \textbf{\#Layers} & \textbf{Batch size} \\ \hline
CartPole-v0          & Quantum          & 0.1                    & 3                 & 10                  \\ \hline
CartPole-v0          & Classical        & 0.01                    & -                 & 10                  \\ \hline
Acrobot-v1           & Quantum          & 0.1                    & 4                 & 10                  \\ \hline
Acrobot-v1           & Classical     & 0.01                    & -                 & 10                  \\ \hline
QControl             & Quantum          & 0.01                   & 1                 & 10                  \\ \hline
QControl           & Classical     & 0.01                    & -                 & 10                  \\ \hline
%MountainCar-v0       & 0.1                    & 4                 & 10                  \\ \hline
\end{tabular}
\vspace{10pt}
\caption{Specification for hyperparameter, number of layers, and batch size used for the classical and quantum neural policies in the three simulated environments.}      
\label{tab: config}
\end{table}
\noindent
 Figures \ref{fig: avg_reward_experiments}(a), \ref{fig: avg_reward_experiments}(b) and \ref{fig: avg_reward_experiments}(c) compare the average cumulative reward through several episodes for quantum and classical neural policies for the Cartpole, Acrobot, and QControl environments, respectively. A running mean was plotted to smooth the reward curves since the policy and environments are noisy. Figure \ref{fig: avg_reward_experiments}(c) also plots the respective control trajectory obtained by the variational quantum policy.

\begin{figure}[!htb]
    \centering
    \includegraphics[width=\textwidth]{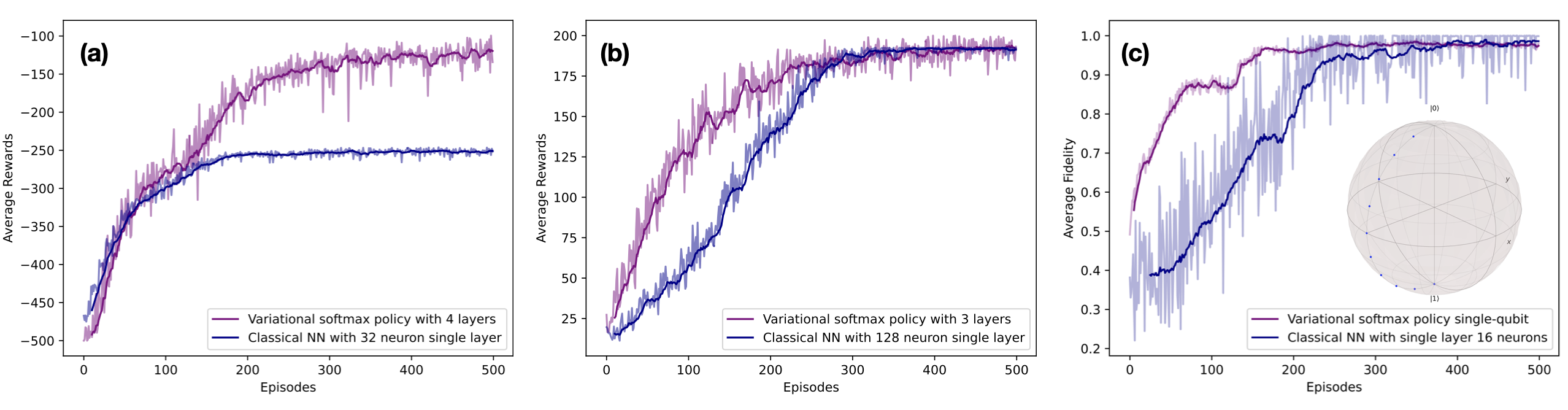}
    \caption{Average cumulative reward. Comparison between the variational softmax policy and the respective classical NN. Panels (a), (b), and (c) represent the average reward comparison for the Cartpole, Acrobot, and QControl environments, respectively.}
    \label{fig: avg_reward_experiments}
\end{figure}
\noindent
 One can conclude that the quantum and classical neural policies perform similarly in every environment. In the QControl environment, the classical policy achieves a slightly greater cumulative reward. Nonetheless, there is clear evidence that the quantum-inspired policy needs fewer interactions with the environment to converge to near-optimal behavior. Moreover, the total number of trainable parameters for the quantum and classical models is summarized in the Table \ref{tab: parameters}. The input layer of a classical neural network is related to the number of qubits in a quantum circuit. Furthermore, we take the number of layers in the VQC as the number of hidden layers in a classical neural network. Given that the quantum circuit is unitary, the number of neurons in a quantum neural network is constant, i.e., equal to the system's number of qubits. Thus, one can conclude that the quantum policy has similar or even outperforming behavior compared to the classical policy with an extremely reduced total number of trainable parameters.\\

\begin{table}[!ht]
\centering
%\begin{tabularx}{\textwidth}{X X X X X X X X}
\begin{tabular}{cccccccc}
\hline
\textbf{Env} & \textbf{Policy} & \textbf{I} & \textbf{O} & \textbf{\#N} & \textbf{\#R} & \textbf{$\beta$} & \textbf{\#P} \\ \hline
CartPole-v0          & Quantum         & 4                    & 2                     & ---              & 2                            & Yes              & 25                    \\ \hline
CartPole-v0          & Classical       & 4                    & 2                     & 128               & ---                         & No               & 768                   \\ \hline
Acrobot-v1           & Quantum         & 6                    & 3                     & ---              & 2                            & Yes              & 33                    \\ \hline
Acrobot-v1           & Classical       & 6                    & 3                     & 32               & ---                          & No               & 288                   \\ \hline
QControl           & Quantum         & 1                    & 1                     & ---              & 3                            & Yes              & 3                    \\ \hline
QControl           & Classical       & 4                    & 2                     & 16               & ---                          & No               & 96                   \\ \hline
\end{tabular}
\vspace{10pt}
\caption{Number of parameters trained for both environments ({\footnotesize \textbf{Env}: environment; \textbf{I}: Input layer; \textbf{O}: Output layer; \textbf{\#N}: neurons; \textbf{\#R}: rotations per qubit; \textbf{\#P}: parameters}).}
\label{tab: parameters}
\end{table}

\subsection{The effect of initialization}\label{subsec: initialization}

	The parameters' initialization strategy can dramatically improve the convergence of a machine learning algorithm. Random initialization is often used to break the symmetry between different neurons \cite{dl_bengio}. However, if the parameters are arbitrarily large, the activation function may saturate, difficulting the learning task. Therefore, parameters are often drawn from specific distributions. For instance, the Glorot  \cite{Glorot2010UnderstandingTD} initialization strategy is among the most commonly used to balance initialization and regularization \cite{dl_bengio}.
\noindent
In quantum machine learning models, the problem persists. However, it was verified experimentally that the Glorot initialization has a slight advantage compared to other strategies. The empirical results reported in Section \ref{sec: benchmarking} were obtained using such a strategy. 
\noindent
The Glorot strategy samples the parameters of the network from a normal distribution $\mathcal{N}(0,std^2)$ with standard deviation given by Equation \eqref{eq: glorot}:
   
	\begin{equation}
	std = gain * \sqrt{6 \over {n_{in} + n_{out}}} 
	\label{eq: glorot}
	\end{equation}
\noindent
where \textit{gain} is a constant multiplicative factor. $n_{in}$ and $n_{out}$ are the number of inputs and outputs of a layer, respectively. It was devised to initialize all layers with approximately the same activation and gradient variance, assuming that the neural network does not have nonlinear activations, being thus reducible to a chain of matrix multiplications. The latter assumption motivates this strategy in quantum learning models since they are composed of unitary layers without nonlinearities. The only nonlinearity is introduced by the measurement \cite{chuang_book}.
\noindent
Figures \ref{fig: vqc_init_experiments}(a), \ref{fig: vqc_init_experiments}(b) and \ref{fig: vqc_init_experiments}(c) plot the average reward obtained by the quantum agent in the CartPole, Acrobot and QControl environments, respectively, following the most common initialization strategies. Glorot initialization has a slightly better performance and stability. Moreover, it is verified empirically that for policy gradients, initialization from normal distributions generates better results for the classic environments compared to uniform distributions, as reported in \cite{Zhang2022} for standard machine learning cost functions. However, in the QControl task was not observed the same behavior since uniform sampling $U(-1,1)$ achieves similar performance than $N(0,1)$.

\begin{figure}[!htb]
    \centering
    \includegraphics[width=\textwidth]{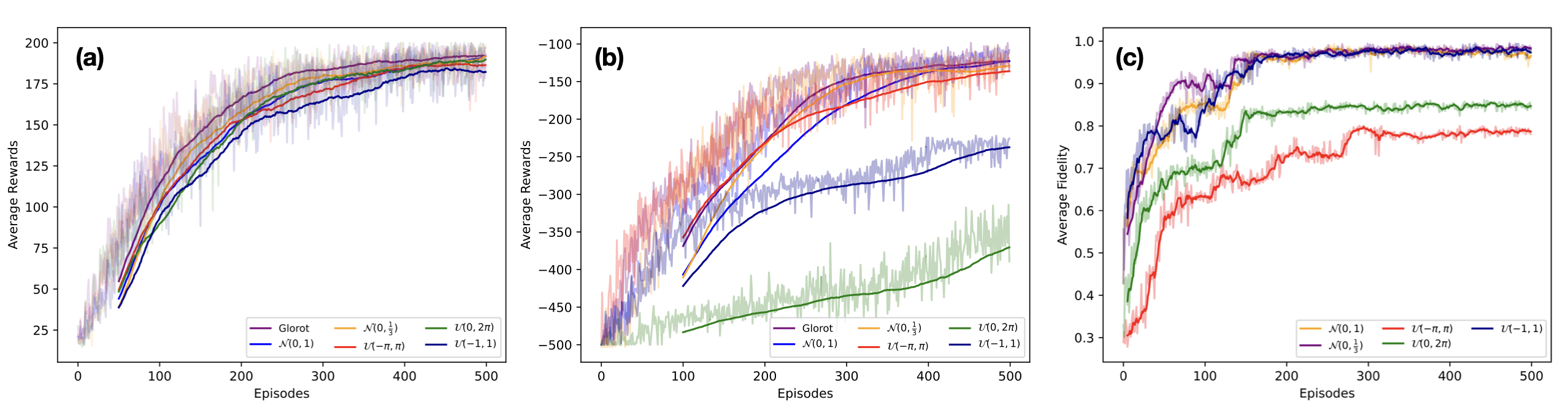}
    \caption{Normal and Uniform distributions used to initialize the parameters of the variational softmax policy. Panels (a), (b), and (c) represent the average reward comparison for the Cartpole, Acrobot, and QControl environments, respectively.}
    \label{fig: vqc_init_experiments}
\end{figure}
\section{Quantum enhancements}\label{sec: advantage}

In this section, further steps are taken toward studying the possible advantages of quantum RL agents following two different strategies:
\begin{itemize}
	\item \textbf{Parameter Count} - Comparison between quantum and classical agents regarding the number of parameters trained. It is unclear whether this is a robust approach to quantify advantage, given that the number of parameters alone can be misleading. For example, the function $sin(\theta)$ has a single parameter and is more complex than polynomial $ax^3 + bx^2 + cx + d$. However, having smaller networks could enable solutions for more significant problems at a smaller cost. Even though only parameter-shift rules are allowed on real quantum hardware, it enables a lower cost on memory than backpropagation. Perhaps the training difference may be negligible from a tradeoff between memory and time consumption for large enough problems. As reported in Table \ref{tab: parameters}, a massive reduction in the number of parameters in the quantum neural network compared with the classical counterpart for all three simulated environments.
	\item \textbf{Fisher Information} - The Fisher Information matrix spectrum is related to the effect of barren plateaus in the optimization surface itself. Studying the properties of the matrix eigenvalues should help to explain the hardness of training.
	\end{itemize}
\noindent	
The Fisher Information \cite{ly2017tutorial} is crucial both in computation and statistics as a measure of the amount of information in a random variable $X$ in a statistical model parameterized by $\theta$. Its most general form amounts to the negative Hessian of the log-likelihood. Suppose a datapoint $x$ sampled i.i.d from $p(x \lvert \theta)$ where $\theta \in \mathbb{R}^k$. Since the Hessian reveals information about the curvature of a function, the Fisher Information Matrix (see Equation \eqref{eq: fisher information matrix}) captures the sensitivity concerning changes in the parameter space, i.e., changes in the curvature of the loss function.

\begin{equation}
 F(\theta) = \mathbb{E}_{x \sim p} \left[ \nabla_{\theta} log p(x \lvert \theta) \nabla_{\theta} log p(x \lvert \theta)^{\top}\right] \in \mathbb{R}^{k \times k}
 \label{eq: fisher information matrix}
\end{equation}
\noindent
The Fisher Information matrix is computationally demanding to obtain. Thus, the empirical Fisher information matrix is usually used in practice and can be computed as in Equation \eqref{eq: empirical fisher information}:

\begin{equation}
 F(\theta) = {1 \over T} \sum_{i=1}^{T} \nabla_{\theta} log p(x_i \lvert \theta) \nabla_{\theta} log p(x_i \lvert \theta)^{\top}
 \label{eq: empirical fisher information}
\end{equation}
\noindent
Equation \eqref{eq: empirical fisher information} captures the curvature of the score function at all parameter combinations. That is, it can be used as a measure for studying barren plateaus in maximum likelihood estimators \cite{karakida2019universal}, given that all the matrix entries will approach zero with the flatness of the model's landscape. This effect is captured by looking at the spectrum of the matrix. If the model is in a barren plateau, then the eigenvalues of the matrix will approach zero \cite{Abbas_2021}.
\noindent
In the context of policy gradients, the empirical Fisher information matrix \cite{kakade_ng} is obtained by multiplying the vector resultant of the gradient of the log-policy with its transpose as in Equation \eqref{eq: pg fisher information}:

\begin{equation}
 F(\theta) = {1 \over T} \sum_{t=1}^T \nabla_{\theta} log \pi(a_t \lvert  s_t,\theta) \nabla_{\theta} log \pi(a_t \lvert  s_t,\theta)^{\top} 
 \label{eq: pg fisher information}
\end{equation}
\noindent
Inspecting the spectrum of the matrix in Equation \eqref{eq: pg fisher information} reveals the flatness of the loss landscape. Thus, it can harness the hardness of the model's trainability for both RL agents based on classical neural networks and VQCs \cite{Abbas_2021}. This work considers the trace and the eigenvalues' probability density of the Fisher Information matrix. The trace will approach zero if the model is closer to a barren plateau and the eigenvalues' probability density unveils the magnitude of the associated eigenvalues.

\begin{figure}[!htb]
    \centering
    \includegraphics[width=\textwidth]{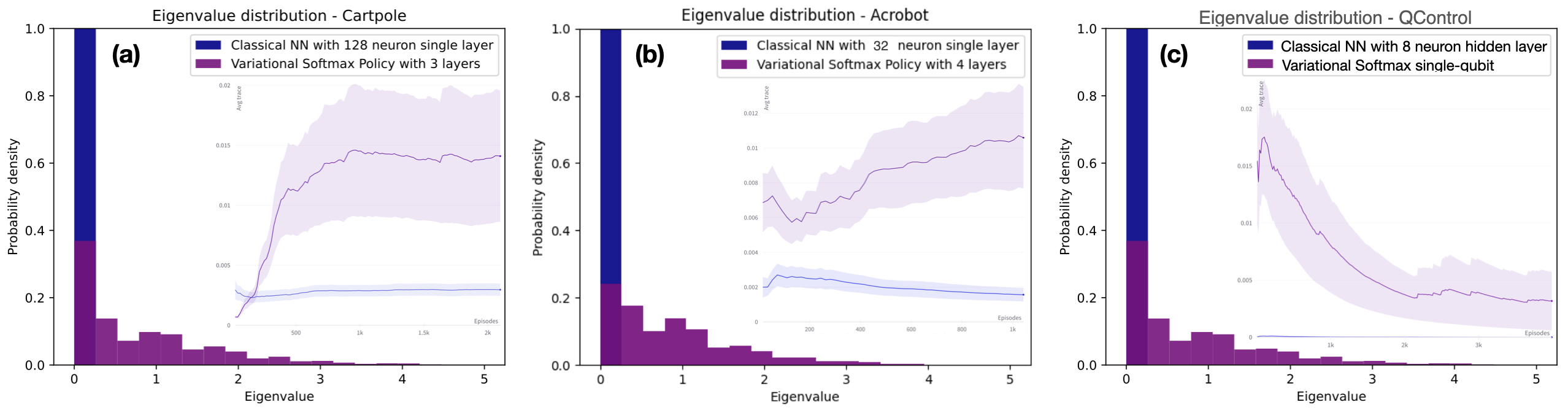}
    \caption{Probability density for the Fisher information matrix eigenvalues and average trace. Panels (a), (b), and (c) represent the eigenvalue distribution and trace of the Fisher information matrix for the Cartpole, Acrobot, and QControl environments, respectively.}
    \label{fig: fisher_experiments}
\end{figure}
\noindent
Figures \ref{fig: fisher_experiments}(a), \ref{fig: fisher_experiments}(b) and \ref{fig: fisher_experiments}(c) plot the average Fisher information matrix eigenvalue distribution for training episodes during the entire training for the CartPole, Acrobot and QControl environments, respectively. Subpanels in every plot indicate the associated information matrix trace. On average, the Fisher information matrix of the quantum model exhibits significantly larger density in eigenvalues different from zero compared to the classical model during the entire training. The same behavior is observed for every environment, explaining the improvement of the training performance for quantum agents (section \ref{sec: performance}) compared to classical ones. Although it is not visible from the eigenvalue distribution, the classical model has larger eigenvalues than the quantum model. However, their density is extremely small, thus making it negligible in a distribution plot. Further analysis is required to understand the behavior of both classical and quantum agents thoroughly.

\section{Conclusion}
\label{sec: conclusion}
In this work, a VQC was embedded into the decision-making process of an RL agent, following the policy gradient algorithm, solving a set of standard benchmarking environments efficiently. Empirical results demonstrate that such variational quantum models behave similarly or even outperform several typically used classical neural networks. The quantum-inspired policy needs fewer interactions to converge to an optimal behavior, benefiting from a reduction in the total number of trainable parameters.\\
Parameter-shift rules were used to perform gradient-based optimization resorting to the same quantum model used to compute the policy.
It was proved that the sample complexity for gradient estimation via parameter-shift rules grows only logarithmically with the number of parameters.\\
The Fisher Information spectrum was used to study the effect of barren plateaus in quantum policy gradients. The spectrum indicates that the quantum model comprises larger eigenvalues than its classical counterpart, suggesting that the optimization surface is less prone to plateaus.\\
Finally, it was verified that the quantum model could prepare a single-qubit state with high fidelity in fewer episodes than the classical counterpart with a single layer.\\
Concerning future work, it would be interesting to apply such RL-based variational quantum models to quantum control problems of larger dimensions. Specifically, their application to noisy environments would be of general interest. Moreover, studying the expectation value of policy gradients given a specific initialization strategy to support empirical claims is crucial. At last, the quantum Fisher Information  \cite{Meyer2021} should be addressed to analyze the information behind quantum states. Moreover, it would be interesting to embed the Quantum Fisher Information in a Natural Gradient optimization \cite{Stokes_2020} to derive Quantum Natural Policy Gradients. Advanced RL models such as Actor-Critic or Deep Deterministic Policy Gradients (DDPG) could benefit from quantum-aware optimization. 

\section*{Acknowledgements} 
This work is financed by National Funds through the Portuguese funding agency, FCT - Fundação para a Ciência e a Tecnologia, within grant LA/P/0063/2020, and project IBEX, with reference PTDC/CC1-COM/4280/2021.

\section*{Declarations}

\subsection*{Conflict of interests} 
The authors declare no competing interests.
\subsubsection*{Data availability} 
Data sharing not applicable to this article as no datasets were generated or analyzed during the current study.
\bibliography{VQPG}% common bib file

\appendix

\section{Upper bounds on gradient estimation}

This appendix develops the proofs for Lemmas \ref{lemma: sample complexity-LPS} and \ref{lemma: query complexity-LPS}, as presented in section \ref{sec: gradient estimation}.

\subsection{$\epsilon_{\nabla}$-approximation of the policy-gradient}\label{appendix: gradient_estimation_LPS}

 Lemma \ref{lemma: sample complexity-LPS} establishes an upper bound on the number of samples required to $\epsilon_{\nabla}$-estimate the policy gradient $\hat{\nabla}_{\theta} J(\theta)$.
 
 \sample*
 
  \begin{proof}
 The policy gradient is estimated by resorting to Monte Carlo techniques, as described by Equation \eqref{eq: quantum policy gradient-LPS}, restated here for completion.
 
 \begin{equation*}
  \nabla_{\theta} J(\theta) \;=\; \frac{1}{N} \sum_{i=0}^{N-1}\sum_{t=0}^{T-1} G_t(\tau_i) \beta \left(\nabla_{\theta} \langle a_{t_i} \rangle_{\theta} - {\sum_{b_{t_i}} \pi(b_{t_i}\lvert s_{t_i},\theta) \nabla_{\theta} \langle b_{t_i} \rangle_{\theta}}\right)
\end{equation*}
\noindent
Recall that the number of samples is defined as the number of visited states. Since there are $N$ trajectories (sequences of actions, $\tau_i$), each visiting $T$ states, the total number of samples is equal to $NT$.\\
Since the expectation value of a single qubit observable is bounded as $\langle \sigma_z \rangle \in \left[-1,1\right]$ and since the gradient of an action's expected value is given by equation \eqref{eq: parameter-shift rules-LPS}, then $\nabla_{\theta} \langle a \rangle_{\theta} \in \left[-1,1\right]$. Therefore, the following holds:
\begin{equation}
\beta \left(\nabla_{\theta} \langle a_{t_i} \rangle_{\theta} - {\sum_{b_{t_i}} \pi(b_{t_i}\mid s_{t_i},\theta) \nabla_{\theta} \langle b_{t_i} \rangle_{\theta}}\right) \in \left[-2\beta , 2\beta \right]
\label{eq: lemma4-1_res1}
\end{equation}
\noindent
By defining $R_{max}$ as the maximum possible reward at any time step and by recalling Equation \eqref{eq: return}, then
\begin{equation}
G(\tau) = \sum_{t=0}^{T-1} \gamma^t r_{t+1} \leq R_{max} \sum_{t=0}^{T-1} \gamma^t = R_{max} \frac{\gamma^T-1}{\gamma-1} 
\end{equation}
where the expression for the sum of T terms of a geometric progression was used. Using this upper bound on $G(\tau)$ enables the following result
\begin{equation}
    \sum_{t=0}^{T-1} G_t(\tau) \leq R_{max} \sum_{t=0}^{T-1} \frac{\gamma^{T-t} - 1}{(\gamma - 1)} \leq R_{max} \frac{T}{(\gamma - 1)^2}
\label{eq: lemma4-1_res2}
\end{equation}
\noindent
where the last inequality can be obtained by algebraic development and by resorting again to the sum of terms of a geometric progression. 
\noindent
Combining results \eqref{eq: lemma4-1_res1} and \eqref{eq: lemma4-1_res2}, gives
\begin{equation}
    \sum_{t=0}^{T-1} G_t(\tau_i) \beta \left(\nabla_{\theta} \langle a_{t_i} \rangle_{\theta} - {\sum_{b_{t_i}} \pi(b_{t_i}\lvert s_{t_i},\theta) \nabla_{\theta} \langle b_{t_i} \rangle_{\theta}}\right) \in \left[-\frac{2\beta R_{max} T}{(\gamma - 1)^2} , \frac{2\beta R_{max} T}{(\gamma - 1)^2} \right]
    \label{eq: lemma4-1_res3}
\end{equation}
\noindent
From Hoeffding's inequality \cite{hoeffding_ineq}, the probability of the average over $N$ estimates of the policy gradient random variable being $\epsilon_{\nabla}$-inaccurate is given by 
 
 %
 % As a reminder that can be useful later, Hoeffding's inequality states that for $Z_1, \ldots, Z_n$ bounded independent random variable, such that $Z_i \in [a,b]$, then $\mathbb{P} \left[ \lvert \frac{1}{n}\sum_{i=1}^N\left(Z_i - \mathbb{E}[Z_i]\right) \rvert \geq \epsilon \right] \leq \exp\left(-\frac{2n\epsilon^2}{(b-a)^2}\right)$
 
 \begin{equation}
 \mathbb{P} \left[ \lvert \nabla_{\theta}^* J(\theta) - \nabla_{\theta} J(\theta)\rvert \geq \epsilon_{\nabla} \right] \leq 2 \exp \left(- \frac{2N\epsilon_{\nabla}^2 (\gamma - 1)^4}{16\beta^2 R_{max}^2 T^2} \right)
 \end{equation}
 \noindent
 From the union bound, for all $k$ parameters , the probability is less than 
 
 \begin{equation}
 \mathbb{P} \left[ \bigcup_{k} 2 \exp \left(- \frac{N\epsilon_{\nabla}^2 (\gamma - 1)^4}{8\beta^2 R_{max}^2 T^2} \right) \right] \leq 2k \exp \left(- \frac{N\epsilon_{\nabla}^2 (\gamma - 1)^4}{8\beta^2 R_{max}^2 T^2} \right)
 \end{equation}
 \noindent
 \vspace{3mm}
Let $\delta_{\nabla}=\mathbb{P} \left[ \lvert \nabla_{\theta}^* J(\theta) - \nabla_{\theta} J(\theta)\rvert \geq \epsilon_{\nabla} \right] $. Then 

\begin{equation}
\begin{split}
1-\delta_{\nabla} & =\mathbb{P} \left[ \lvert \nabla_{\theta}^* J(\theta) - \nabla_{\theta} J(\theta)\rvert \leq \epsilon_{\nabla} \right] \\
1-\delta_{\nabla} & \geq 1 - 2k \exp \left(- \frac{N\epsilon_{\nabla}^2 (\gamma - 1)^4}{8\beta^2 R_{max}^2 T^2} \right) \\
\delta_{\nabla}& \leq 2k \exp \left(- \frac{N\epsilon_{\nabla}^2 (\gamma - 1)^4}{8\beta^2 R_{max}^2 T^2} \right)
\end{split}
\end{equation}
 \noindent
 Thus, an upper bound on $N$ can be obtained
 \begin{equation}
 N \leq \frac{8\beta^2 R_{max}^2 T^2}{\epsilon_{\nabla}^2 (\gamma - 1)^4} \log ({2k \over \delta_{\nabla}})
 \end{equation}
 \noindent
 \noindent
Considering $NT$ samples completes the proof.
 
\end{proof}

\subsection{Total number of quantum circuit evaluations}\label{appendix: query_complexity_LPS}
 
Lemma \ref{lemma: query complexity-LPS} establishes an upper bound on the number of quantum circuit evaluations (or shots) required to $\epsilon_{\langle \rangle}$-estimate the policy gradient $\hat{\nabla}_{\theta} J(\theta)$ with probability $1-\delta_{\langle \rangle}$. This result builds on Lemma \ref{lemma: sample complexity-LPS} and the same approach is used to demonstrate it.
 
 \query*

\begin{proof}

An action preference observable $\langle a \rangle_{\theta}$ is given by a single-qubit observable $\langle \sigma_z \rangle$, as described in Section \ref{sec: measurements}. The number of shots, $n'$, required to estimate the observable expectation with additive error $\epsilon_{\langle \rangle}$ with probability $1-\delta_{\langle \rangle}$ is akin to the estimation of the probability of a Bernoulli distribution using Hoeffding inequality. Since $\langle a \rangle_{\theta} \in [-1,1]$, then, by resorting to Hoeffding inequality and the union bound, we have
 \begin{equation}
 \mathbb{P} \left[ \lvert \langle a \rangle_{\theta}^* - \langle a \rangle_{\theta} \rvert \geq \epsilon_{\langle \rangle} \right] \leq 2 k \exp \left(- \frac{n'\epsilon_{\langle \rangle}^2}{2} \right)
 \end{equation}
 \noindent
\noindent
Following the same reasoning as described in the proof of Lemma \ref{lemma: sample complexity-LPS} , $n'$ is given by
\begin{equation}
 n' \leq \frac{2}{\epsilon_{\langle \rangle}^2} \log (\frac{2k}{\delta_{\langle \rangle}})
 \end{equation}
 \noindent
Since the observable's gradient $\nabla_{\theta}\langle a \rangle_{\theta}$ is estimated via parameter shift rules, as stated in Equation \eqref{eq: parameter-shift rules-LPS}, it requires the estimation of each action preference observable twice, i.e. both $\langle a \rangle_{\theta+\frac{\pi}{2}}$ and $\langle a \rangle_{\theta-\frac{\pi}{2}}$.  Therefore, the number of shots, $n$, required to estimate $\nabla_{\theta}\langle a \rangle_{\theta}$ is given by
\begin{equation}
 n = 2n' \leq \frac{4}{\epsilon_{\langle \rangle}^2} \log (\frac{2k}{\delta_{\langle \rangle}}) \approx \mathcal{O}\left( \frac{1}{\epsilon_{\langle \rangle}^2} \log (\frac{2k}{\delta_{\langle \rangle}}) \right)
 \end{equation}
 \noindent
Recalling that $\mathcal{O}(NT)$ samples are needed as in Lemma \ref{lemma: sample complexity-LPS} and that each sample incurs $\lvert A \rvert$ estimates, completes the proof.

\end{proof}

\end{document}